\documentclass{lmcs} 


\usepackage[utf8]{inputenc}
\usepackage[T1]{fontenc}
\usepackage{amssymb,stmaryrd,amsthm,amsmath,color}
\usepackage{libertine,listings}
\usepackage{enumerate}
\usepackage{xspace}
\usepackage{bm}

\usepackage{stmaryrd}
\SetSymbolFont{stmry}{bold}{U}{stmry}{m}{n} 

  \renewcommand{\nc}[1]{\newcommand{#1}}
  \newcommand{\rnc}[1]{\renewcommand{#1}}

  \nc{\cal}[1]{\ensuremath{\mathcal{#1}}}

  \nc{\e}{\epsilon}
  \nc{\ez}{\delta}
  \nc{\ph}{\varphi}

  \nc{\w}{\overline{w}}
  \nc{\x}{\overline{x}}
  \nc{\y}{\overline{y}}
  \nc{\z}{\overline{z}}
  \rnc{\a}{\overline{a}}
  \rnc{\b}{\overline{b}}
  \rnc{\c}{\overline{c}}
  \rnc{\d}{\overline{d}}
  \nc{\ov}[1]{\ensuremath{\overline{#1}}}

  \nc{\A}{\mathcal{A}}
  \nc{\B}{\mathcal{B}}
  \nc{\C}{\mathcal{C}}
  \nc{\D}{\mathcal{D}}

  \nc{\X}{\mathcal{X}}
  \nc{\G}{\mathcal{G}}
  \nc{\N}{\mathcal{N}}
  \nc{\T}{\ensuremath{{\cal T}}}
  \rnc{\H}{\mathcal{H}}

  \nc{\set}[1]{\ensuremath{\{#1\}}}
  \nc{\Set}[1]{\ensuremath{\big\{#1\big\}}}
  \nc{\setc}[2]{\set{#1 \ | \ #2}}
  \nc{\Setc}[2]{\Set{#1 \ \big| \ #2}}
  \nc{\size}[1]{|\!|#1|\!|}
  \nc{\setsize}[1]{\ensuremath{|#1|}}
  \nc{\Setsize}[1]{\ensuremath{\big|#1\big|}}

  \nc{\NN}{\ensuremath{\mathbb{N}}}
  \nc{\NNpos}{\ensuremath{\NN_{\scriptstyle\geq 1}}}
  \nc{\RR}{\ensuremath{\mathbb{R}}}
  \nc{\RRpos}{\ensuremath{\RR_{\scriptstyle > 0}}}

  \nc{\nul}{\textsf{Null}}
  \nc{\dist}{\mathrm{dist}}
  \nc{\dst}[1]{\ensuremath{\mathrm{dist}_{\le #1}}}
  \nc{\FO}{\ensuremath{\textup{FO}}\xspace}

  \nc{\sk}{\ensuremath{{\textup{SKIP}}}}
  \nc{\Iskip}{\ensuremath{S}} 
  \nc{\Jskip}{\ensuremath{S'}} 
  \nc{\SC}{\ensuremath{\textit{SC}}}

  \nc{\AWstar}{\ensuremath{\textup{AW[$*$]}}}

\theoremstyle{plain}
	\newtheorem{cla}[thm]{Claim}
\theoremstyle{definition}
	\newtheorem{defn}[thm]{Definition}

\begin{document}

\title[]{Dynamic Query Evaluation Over Structures With Low 
Degree}

\author[A.~Vigny]{Alexandre Vigny}	
\address{Universität Bremen, Germany}	
\email{vigny@uni-bremen.de}  
\thanks{Big thanks to the anonymous reviewers who read the firsts versions of that paper. Special thanks to Luc Segoufin and Sophie Dramé-Maigné for helping me with their thoughtful comments}	

%





\begin{abstract}
  \noindent We consider the evaluation of first-order queries over classes of databases that have \emph{bounded degree} and \emph{low degree}.
	More precisely, given a query and a database, we want to efficiently test whether there is a solution, count how many solutions there are, or be able to enumerate the set of all solutions.

  Bounded and low degree are rather natural notions and both yield efficient algorithms.
	For example, Berkholz, Keppeler, and Schweikardt showed in 2017 that over databases of bounded degree, not only any first order query can efficiently be tested, counted and enumerated, but the data structure used can be updated when the database itself is updated.

	This paper extends existing results in two directions. First, we show that over classes of databases with low degree, there is a data structure that enables us to test, count and enumerate the solutions of first order queries. This data structure can also be efficiently recomputed when the database is updated. Secondly, for classes of databases with bounded degree we show that, without increasing the preprocessing time, we can compute a data structure that does not depend on the query but only on its quantifier rank. We can therefore perform a single preprocessing that can later be used for many queries.
\end{abstract}

\maketitle


\section{Introduction}

Query evaluation is one of the most central tasks of a database system and a vast amount of literature is devoted to the complexity of this problem. Given a database $\D$ and a query $\ph$, the goal is to compute the set $\ph(\D)$ of all solutions for $\ph$ over $\D$. Unfortunately, the set $\ph(\D)$ might be much bigger than the database itself, as the number of solutions may be exponential in the arity of the query.  It can therefore be insufficient to measure the complexity of answering $\ph$ on $\D$ only in terms of the total time needed to compute the complete result set $\ph(\D)$.  Many other algorithmic tasks have been studied to overcome this situation. We can for instance only compute the number of solutions or simply test whether there exists a solution.

In addition, we consider evaluation of queries in a ``dynamic setting'', where the databases may be updated by inserting or deleting tuples. In this setting, given a database $\D$ and a query $\ph$, the goal is to compute a \emph{data structure} that represents the result of evaluating $\ph$ on $\D$. After every database update, the data structure must be adapted so that it represents the result of evaluating $\ph$ on $\D'$, the new database.

Since we do not wish to compute the entire set of solutions, we want the data structure to enable us to efficiently perform the following tasks:
  \begin{itemize}
    \item \emph{Model check:} Decide whether $\ph(\D)=\emptyset$.
    \item \emph{Enumerate:} Enumerate the set $\ph(\D)$ i.e.~to output one by one all elements of $\ph(\D)$.
    \item \emph{Test:} Given a tuple $\a$, decide whether it is in $\ph(\D)$.
    \item \emph{Count:} Compute the number $\#\ph(\D)$ of solutions.
  \end{itemize}

The time needed to solve these problems is split between three sub-processes.
The first one is the \emph{preprocessing}, i.e.~the process that computes the data structure. The second one is the \emph{update}, i.e.~the process that adapts the data structure when the database is updated. The third and final process is the \emph{answering}, or \emph{enumeration}. For the enumeration problem, the {\em delay} is the maximum time between the outputs of any two consecutive solutions.

We want to perform these processes as efficiently as possible.
For the answering or enumeration process, the best we can hope for is to perform it in constant time (or with constant delay): depending only on the query and independent from the size of the database. For the preprocessing, an ideal time would be linear time: linear in the size of the database with a constant factor depending on the query. The ideal update time depends on what kind of updates are allowed. If we only allow \emph{local updates} that only affect a fixed number of tuples, one can hope for update procedures working in constant or almost constant time.

\subsection{The need of restrictions}
Efficient algorithmic procedures for these problems cannot be achieved for all queries over all databases. Even without updates the model checking of first-order ($\FO$) queries over arbitrary relational structure is, by definition, $\AWstar$-complete~\cite[Definition~8.29]{DBLP:series/txtcs/FlumG06}. But for restricted classes of queries and databases, several evaluation algorithms with linear time preprocessing have been obtained, first without updates.

This is the case for arbitrary databases, when paired with free-connex acyclic conjunctive queries~\cite
{DBLP:conf/csl/BaganDG07}, some unions of conjunctive queries~\cite
{DBLP:conf/pods/CarmeliK19}, or conjunctive queries with bounded submodular width~\cite
{DBLP:conf/mfcs/BerkholzS19}.
Allowing for more powerful query languages requires restrictions on the database part. Evaluation algorithms with linear preprocessing exist for
$\FO$ queries over classes of databases with bounded degree~\cite{DBLP:journals/tocl/DurandG07,DBLP:journals/corr/abs-1105-3583}, bounded expansion~\cite{DBLP:conf/pods/KazanaS13}, and for monadic second-order (MSO) queries over classes of databases of bounded tree-width~\cite{DBLP:conf/csl/Bagan06,DBLP:journals/tocl/KazanaS13}.

In some scenarios only pseudo-linear preprocessing time has been achieved. A query can be enumerated with constant delay after pseudo-linear preprocessing if for all $\epsilon$ there exists an enumeration procedure with constant delay (the constant may depend on $\epsilon$) and preprocessing time in $O(\size{\D}^{1+\epsilon})$, where $\size{\D}$ denotes the size of the database. This has been achieved for $\FO$ queries over classes of databases with low degree~\cite{DBLP:conf/pods/DurandSS14}, locally bounded expansion~\cite{DBLP:conf/icdt/SegoufinV17}, and for nowhere dense classes of databases~\cite{DBLP:journals/jacm/GroheKS17,DBLP:journals/corr/GroheS17,DBLP:conf/pods/SchweikardtSV18}.

If we allow updates, the problem is at least as hard.
This more restricted line of research is quite recent but already some dynamic procedures have been provided: for conjunctive queries over arbitrary relational structures~\cite{DBLP:conf/pods/BerkholzKS17,DBLP:conf/icdt/BerkholzKS18}, for MSO queries over words, trees, and graphs with bounded tree-with~\cite{DBLP:conf/csl/LosemannM14,
DBLP:conf/icdt/AmarilliBM18,DBLP:conf/pods/NiewerthS18}, and for \FO queries over graphs with bounded degree~\cite{DBLP:journals/tods/BerkholzKS18} and bounded expansion~\cite{DBLP:conf/pods/Torunczyk20}.

Interested readers can look at the third chapter (Figure 3.1) of the Thesis~\cite{DBLP:phd/hal/Vigny18}, or Figures 4 and 5 of~\cite{DBLP:conf/pods/0002NOZ20} for a more detailed overview of these results.

\subsection{Main results}

Our first result concerns classes of databases with low degree. We show that for such classes, and for any first order query, there is a data structure that can be computed in time $O(n^{1+\e})$ (where $n$ is the size of the database) and recomputed in time $O(n^\e)$ when the database is updated.
Using the data structure, the answering processes for the model-checking, testing and counting problems can be performed in constant time. Additionally, the set of solutions can be enumerated with constant delay, without further preprocessing.
See Theorems~\ref{thm-low-mc} and~\ref{thm-low-all} for formal statements of that result.

The second result concerns classes of graphs with bounded degree. We expand on existing results and show that a single linear time preprocessing can be used to model-check, test, and count in constant time, and also to enumerate with constant delay the solutions of a whole family of queries. We are also able to update the results in constant time when the database is updated.
We already knew that it was possible when only one query is given as input~\cite{DBLP:journals/tods/BerkholzKS18}.
Additionally, one can note that if a finite family of queries is given as input, it is easy to perform one preprocessing for each query. The parameter would simply be the sum of the sizes of the queries.
However, in our contribution,
the preprocessing time is smaller than the time it would take to perform one preprocessing for each individual query of the inputed family. Even better, our running time matches those of the algorithms that take as input only one query: triply exponential in the quantifier rank of the query, or in our case in the maximal quantifier rank amongst the possible queries.
See Theorems~\ref{thm-bound-mc} and~\ref{thm-bound-all} for formal statements of that result.

\subsection{Related techniques}\label{seq-techniques}
While the results are new, the techniques used to obtain them are not. We are using Gaifman local form of first order formulas, which is already at the core of the existing results for structures with bounded and low degree~\cite{DBLP:journals/tods/BerkholzKS18,DBLP:conf/pods/DurandSS14,DBLP:journals/corr/abs-1105-3583}.

Enumeration algorithms with constant delay often need a mechanism used to skip irrelevant parts of the input.
An adequate technique (with small variations) has been used several times in previously published results. The firsts appearances where in~\cite[Section~3.5]{DBLP:conf/pods/DurandSS14}
and~\cite[Section~4]{DBLP:conf/pods/KazanaS13}. Some variations have been used again in~\cite[Lemma 10.1]{DBLP:journals/tods/BerkholzKS18} and~\cite[Lemma 3.10]{DBLP:conf/pods/SchweikardtSV18}. We use a similar structure described in Lemma~\ref{lem-skip}.
Concerning the counting part of our results, we use a classic inclusion-exclusion argument. This argument is materialized in Lemma~\ref{lemma-count-decomposition}, which is analogous to what can be found in~\cite[Section~7]{DBLP:conf/icdt/SegoufinV17}.

The proofs of our main results follow the same path. In both cases we focus first on centered queries (Definition~\ref{defn-centered}). These queries only talk about the neighborhood of one node. When the degree of the database is small, that neighborhood is also small. It is therefore possible to entirely compute all solutions of a given centered query in (pseudo) linear time. The main results are then proved using the fact that all first order queries are combinaisons of centered ones together with adequate tools like Lemmas~\ref{lem-skip} and~\ref{lemma-count-decomposition}.

\subsection{Organization}
The rest of the paper is structured as follows.
Section~\ref{seq-prelim} provides basic notations and additional information about our setting and our computational model.
In Section~\ref{seq-tech-def} we state refinements of the results our paper is build on, mainly, Gaifman locality of first-order queries.
Finally, Section~\ref{seq-low-degree} and~\ref{seq-bound} prove the two main results respectively.


\section{Preliminaries}\label{seq-prelim}

By $\NN$ we denote the non-negative integers, and we let $\NNpos := \NN\setminus\set{0}$. Throughout this paper, $\e$ will always be a positive real number, and $d$, $i$, $j$, $k$, $\ell$, $m$, $n$, $p$, $q$, $r$ will be elements of $\NN$.

\subsection{Colored graphs and first-order queries}

We define \emph{undirected colored graphs} (or simply graphs) as sets of vertices equipped with a symmetric binary relation (edges) and some unary predicates (colors).

As mentioned in the introduction, our algorithms work for classes of databases. However we will only prove results for classes of graphs. Our algorithms can be adapted to the general case of databases (viewed as relational structure) via standards techniques. Namely, one can take any relational structure and compute a colored version of its adjacency graph.
Interested readers can look at~\cite{DBLP:conf/pods/SchweikardtSV18} {\it (Section~2, From databases to colored graphs)} to see the details of this transformation.

We denote by $|G|$ the size $|V|$ of its vertex set, and $\|G\|$ the size of (an encoding of) $G$. Without loss of generality we assume that the set $V$ comes with a linear order. If not we arbitrarily choose one. This order induces a lexicographical order among the tuples of vertices.

In this paper, queries are first-order formulas. We assume familiarity with first-order logic ($\FO$), over relational structures (see~\cite{DBLP:books/aw/AbiteboulHV95,DBLP:books/sp/Libkin04} for more information). We use standard syntax and semantics for \FO. In particular we write $\ph(\x)$ to denote the fact that the free variables of the query $\ph$ are exactly the variables in $\x$. The length of $\x$ is called the \emph{arity} of the query. The \emph{size} of a query $\ph$ is the number of symbols needed to write down the formula and is denoted by $|\ph|$, and the \emph{quantifier rank} of a query $\ph$ is the depth of nesting of its quantifiers.

A \emph{sentence} (or boolean query) is a formula with no free variable, i.e. of arity~0.

We write $G \models \ph(\a)$ to indicate that $\a$ is a solution for $\ph$ over $G$ and note $\ph(G)$ the set of tuples $\a$ such that $G \models \ph(\a)$.

\subsection{Model of computation}

As customary when dealing with linear time, we use Random Access Machines (RAM) with uniform cost measure as a model of computation. More precisely, we need our computational model to perform two tasks.
The first one is to store (ordered) lists with the ability of testing whether an element is in the list in constant time. The ordered version would require logarithmic time to insert (resp. remove) an element to (resp. from) the list.
The second one is to store the values of a given function, and later retrieve the image of a given input in constant time. These requirements are standards, see e.g.~\cite{DBLP:conf/pods/BerkholzKS17} {\it (Section~2, Dynamic Algorithms for Query Evaluation)}.

Both come for free in the RAM model. For any fix integer $k$, it is possible to associate any $k$-tuple of integers to a specific cell of the memory and, given that tuple, to access this cell in constant time. This together with the fact that we can assume that the memory is initially empty (which require a technique called the lazy array initialisation~\cite{moret1991algorithms}) enable us to perform these two tasks.

While this works on the RAM model, it shows that we are getting away from real-life computers.

\subsection{Parametrized complexity}

All problems encountered in this paper have two inputs: a graph $G$ and a query $\ph(\x)$. However, they play different roles as $\size{G}$ is often very large while $|\ph|$ is generally small. We therefore adopt the parameterized complexity point of view. When we say \emph{linear time} we mean in time $O(\size{G})$, the constants hidden behind the ``big $O$'' depending on $\ph$, on the class $\C$ of graphs under investigation, and possibly on further parameters that will be clear from the context.

We say that a problem is solvable in \emph{pseudo-linear time} if, for all $\e$, it can be solved in time $O(\size{G}^{1+\epsilon})$. In this case, the constant factor also depends on $\e$. If a subroutine of a procedure depending on $\e$ produces an output of size $O(\size{G}^{\e})$ we will say that the output is \emph{pseudo-constant}.

\subsection{Distance and neighborhoods}

Given a graph $G=(V,E)$ and two elements $a$ and $b$ of $V$, the \emph{distance} between $a$ and $b$ noted $\dist(a,b)$ is the length of the shortest path between $a$ and $b$ (the number of edges). If there is no such path, $\dist(a,b):=\infty$. The notion of distance extends to tuples in the usual way, i.e., the distance between two tuples $\a$ and $\b$ is the minimum of the distances between $a_i$ and $b_j$ over all $i,j$.

For a positive integer $r$, we write $N_r^G(a)$ for the set of all elements of $V$ at distance at most $r$ from $a$. The $r$-neighbor\-hood of $a$ in $G$, denoted $\N^G_r(a)$, is the substructure of $G$ induced by $N_r^G(a)$. Similarly, for a tuple $\a$ of arity $k$ we let $N_r^G(\a):=\bigcup_{i\le k}N_r^G(a_i)$, and we define $\N_r^G(\a)$ as the substructure of $G$ induced by $N_r^G(\a)$.

\subsection{Graphs with bounded and low degree}

The \emph{degree} of an undirected graph $G=(V,E)$ is the maximal number, over all elements $a$ of $V$, of elements $b$ in $V$ such that $(a,b)\in E$. We note $d(G)$ the degree of $G$.

A class $\C$ of graphs has \emph{bounded degree} if there is an integer $d$ such that every graph of $\C$ has degree at most $d$.

A class $\C$ of graphs has \emph{low degree} if there is a function $f(\cdot)$ such that for any $\e>0$, and for any graph $G$ of $\C$, we have $d(G) \le f(\e)\cdot|G|^\e$. Examples of classes of graphs with low degree are classes of graphs with bounded degree and classes of graphs $G$ with $d(G)\le\log(|G|)$.

\subsection{Updates}

We allow graphs to be updated. A \emph{local update} of a graph $G$ is a graph $G'$ obtained from $G$ by inserting or deleting a node, an edge, or by modifying the membership of a given node to a unary predicate (relabeling). In this paper, we restrict our attention to classes of graphs with bounded or low degree. Therefore when updating a graph, the update is allowed only if the bounds on the degree are maintained in the new graph.

For classes of graphs with bounded degree, given a number $d$ it is easy to check in time $O(d)$ whether when adding a given edge to a graph of degree at most $d$, we make the new graph to have a degree exceeding $d$ or not.

For classes of graphs with low degree, the problem is a little bit trickier. A class of graph $\C$ has \emph{effectively low degree} if there is a {\em computable} function $f(\cdot)$ such that for every $\e>0$, and for all graph $G$ in $\C$, we have $d(G)\le f(\e)\cdot|G|^\e$.


If $\C$ has effectively low degree, then it is possible (for example) to maintain the ordered list of every vertex's degree in time $\log(|G|)$. Then for any $\e>0$, it is possible to check in time $\log(|G|) = O(f(\e)\cdot|G|^\e)$ whether the degree of $G$ still remains below that bound. In this paper, we only talk about classes of graphs with effectively low degree and therefore often omit the ``effectively'' for better readability.
When dealing with update, the following definition and lemma (although quite elementary) simplify every proof.
\begin{defn}
  Let $G$ be a graph and $G'$ a local update of $G$. For an integer $r$ in $\NN$, a node $a$ is \emph{$r$-impacted} if it is: (i) only in $G$, or (ii) only in $G'$, or (iii) in both but $\N^G_r(a) \neq \N^{G'}_r(a)$.
\end{defn}

\begin{lem}\label{lem-update}
  Let $G$ be a graph and $G'$ a local update of $G$. For any integer $r$ in $\NN$, the list $L_u$ of all $r$-impacted nodes has size $O(d^r)$ and can be computed in time $O(d^r)$, where $d$ is the degree of $G$.
\end{lem}

\section{Technical definitions}\label{seq-tech-def}
In this section we explain and define some tools that are used throughout the rest of the paper.
We first define the notions of distance queries, distance types, and later explain how they are used to express \FO queries in a local and succinct way.

\subsection{Distance and types}

A distance query is a binary query that only talks about the distance of two nodes in a graph. For all $r\in \NN$, the query $\dist_{\le r}(a,b)$ states that there is a path of length at most $r$ between $a$ and $b$. It can formally be defined as follows:
$$\dst{0}(x,y) :=  (x=y)  \quad\quad \dst{r+1}(x,y) :=   \dst{r}(x,y) ~\vee ~ \exists z ~ \big( E(x,z) \wedge \dst{r}(z,y) \big).$$
Once this is defined, one can also make use of those variations: $\dist_{>r}(x,y) := \neg\dst{r}(x,y)$ and
$\dist_{=r}(x,y) := \dst{r}(x,y) \wedge \dist_{>r-1}(x,y)$.

\begin{defn}[Distance type]
  For a graph $G=(V,E)$ and an integer $r$ in $\NN$, the \emph{$r$-distance type} of a tuple $\a=(a_1,\ldots,a_k)$ is the undirected graph $\tau_r^G(\a)$ whose nodes are $[1,k]$ and edges are exactly the pairs $(i,j)$ such that $G\models\dst{r}(a_i,a_j)$. We denote by ${\cal T}_k$ the set of all possible distance types with $k$ elements.
\end{defn}

For every integers $r$, $k$ and distance type $\tau\in{\cal T}_k$, there is a query $\Delta_{r,\tau}(\x)$ such that for every graph $G$ and tuple $\a$, we have that $G\models \Delta_{r,\tau}(\a)$ if and only if $\tau_r^G(\a) = \tau$. Furthermore $\Delta_{r,\tau}(\x)$ is a boolean combination of distance queries.

For example, if $k=3$ and $\tau$ only contains the edge $(1,2)$, then $\Delta_{r,\tau}(x,y,z) := \dst{r}(x,y)\wedge\dist_{>r}(x,z)\wedge\dist_{>r}(y,z)$.

\subsection{Normal form for $\FO$ queries}

The classes of graphs we consider have
small neighborhoods. It is therefore interesting to use the locality of first order logic.

\begin{defn}[Local queries]\label{defn-local}
  A query $\ph(\bar x)$ is said to be $r$-local if all its quantifications are relative to elements at distance at most $r$ from one of its free variables $\bar x$,
  i.e.~if it uses quantifications of the form $\exists y~ \dst{r}(y,\bar x) \land \cdots$ and $\forall y~ \dst{r}(y,\bar x) \Rightarrow \cdots$ .
\end{defn}

\begin{thm}[Gaifman's theorem. Originally~\cite{Gaifman1982105}, see~{\cite[Theorem 4.22]{DBLP:books/sp/Libkin04}}]\label{thm-nf-mc}
  For every $\FO$ formula $\phi(\x)$, there are $r,s \in \mathbb{N}$ such that $\phi$ is equivalent to a boolean combination of $r$-local queries and sentences of the form:

  \begin{equation}\label{eq-nf-mc}
  \exists x_1\ldots x_s ~\Big( \underset{1\le i \le j \le s}{\bigwedge}~\dist_{>2r}(x_i,x_j) \wedge \underset{1\le i\le s}{\bigwedge} \alpha(x_i) \Big)
  \end{equation}

  Here $\alpha(x)$ is an $r$-local query. Moreover, this combination is computable from $\phi$, $r = 2^{O(q)}$, and $s\le q$ where $q$ is the quantifier rank of $\ph$.
\end{thm}

When working with non boolean queries, locality is not enough. We also want to know how the free variables are scattered across the graph.

\begin{defn}[Local centered queries]\label{defn-centered}
  An $r$-local query $\alpha(\x)$ is said to be \emph{$r$-local centered} if for every graph $G$ and tuple $\a$, if $G\models\alpha(\a)$ then the $2r$-distance type of $\a$ has only one connected component.
\end{defn}
Note that an $r$-local centered query of arity $k$ is also $2kr$-local around any of its free variables.
It is also worth mentioning that an $r$-local query that is unary is necessarily centered. Similarly, the query $\dst{2r}(x,y)$ is $r$-local centered, and $\Delta_{2r,\tau}(\x)$ is $r$-local centered if and only if $\tau$ has exactly one connected component.

\begin{cor}[Of Theorem~\ref{thm-nf-mc}]\label{cor-fo-queries}
  For every integers $k$ and $q$, there is an integer $r$ such that every query $\ph(\x)$ of arity $k$ and quantifier rank $q$ is equivalent to a boolean combination of sentences of the form of Equation~\ref{eq-nf-mc} and $r$-local queries. Furthermore, once all sentences have been replaced by their truth value, we are left with a strict disjunction of queries of the form:
  \begin{equation}\label{eq-nf-k} \alpha_1(\x_{1})\wedge \ldots\wedge \alpha_p(\x_{p}) \wedge \Delta_{2r,\tau} (\x),\end{equation}
    where $r=2^{O(q)}$, $\alpha_i(\x_{i})$ are $r$-local centered queries, $(I_1,\ldots,I_p)$ is a partition of $[k]$ and $\tau$ is a graph whose connected components are exactly $(I_1,\ldots,I_p)$, with $I_i :=\{j\le k ~:~ x_j \in \x_i \}$.
\end{cor}

To obtain this corollary from Theorem~\ref{thm-nf-mc}, we need (among other things) to take a disjunction over all possible $\tau$ in ${\cal T}_k$. For each $\tau$, the remaining local query can be simplified as all atoms of the form $E(x_i,x_j)$ can 
be set to false, as soon as $i$ and $j$ are not in the same connected component of $\tau$.

\section{Structures with low degree}\label{seq-low-degree}
This section is devoted to the proof of our first contribution. 
\begin{thm}\label{thm-low-mc}
  For every class of graphs $\C$ with effectively low degree, there is a computable function $f$ and an algorithm that upon input of a graph $G$ in $\C$, an $\FO$ sentence $\phi$ and an $\e >0$, computes a data structure that:
  \begin{itemize}
    \item can be computed in time $O\big(f(|\phi|) \cdot |G|^{1+\e}\big)$,
    \item can be recomputed in time $O\big(f(|\phi|) \cdot |G|^\e\big)$ when $G$ is subject to a
    local update, and
    \item enables us to test in time $O(1)$ whether $G\models \phi$.
  \end{itemize}
\end{thm}
This theorem is proved in Section~\ref{subseq-low-mc}.

\begin{thm}\label{thm-low-all}
  For every class of graphs $\C$ with effectively low degree, there is a computable function $f$ and an algorithm that upon input of a graph $G$ in $\C$, an $\FO$ query $\ph(\x)$ and an $\e >0$, computes a data structure that:
  \begin{itemize}
    \item can be computed in time $O\big(f(|\ph|) \cdot |G|^{1+\e}\big)$, and
    \item can be recomputed in time $O\big(f(|\ph|) \cdot |G|^\e \big)$ when $G$ is subject to a local update.
    \item Furthermore, the data structure enables us:
    \begin{enumerate}
      \item~{\textit{\small (enum:)}} to enumerate with delay $O\big(f(|\ph|)\big)$, and without further preprocessing the set $\ph(G)$ of solutions for $\ph(\x)$ over $G$,
      \item~{\textit{\small (test:)}} upon input of a tuple $\a$, to test in time $O\big(f(|\ph|)\big)$ whether $G\models \ph(\a)$, and
      \item~{\textit{\small (count:)}} to compute in time $O\big(f(|\ph|)\big)$ the number $\#\ph(G)$ of solutions for $\ph$ over $G$.
    \end{enumerate}
  \end{itemize}
\end{thm}


The proof of these two Theorems follow standard paths. We show in the next section that the Theorems hold for $r$-local centered queries (see Definition~\ref{defn-centered}). We then generalize to full \FO with combinatorial techniques.
As explained in Section~\ref{seq-techniques}, these techniques are similar to what can be found in the static case~\cite{DBLP:conf/pods/DurandSS14} or the dynamic case for structures with bounded degree~\cite{DBLP:journals/corr/abs-1105-3583,DBLP:journals/tods/BerkholzKS18}.


\subsection{Centered queries}\label{subseq-low-unary}

We restrict our attention to $r$-local centered queries. For these queries, and thanks to the degree bound,
it is actually possible to compute the entire result and to recompute it when the graph is updated.
\begin{lem}\label{lem-low-unary}
    Let $\C$ be a class of graphs with effectively low degree. For every graph $G$ in $\C$ with $n:=|G|$, every $\e>0$, every $r\in \NN$ and any $r$-local centered query $\alpha(\x)$, we can:
  \begin{itemize}
    \item compute $L_\alpha$ and $\#\alpha$ in time $O(n^{1+\e})$, and
    \item recompute $L_\alpha$ and $\#\alpha$ in time $O(n^\e)$ whenever a local update occurs in $G$.
  \end{itemize}
  Here, $L_\alpha$ is the ordered list of tuples from $G$ satisfying $\alpha(\x)$. 
  By $\#\alpha$ we mean $|\{\a~:~ G\models \alpha(\a)\}|$, which is also $|L_\alpha|$.
\end{lem}

\begin{proof}[Proof of Lemma~\ref{lem-low-unary}]
  The challenging part of the proof of Lemma~\ref{lem-low-unary} is the update part. The preprocessing one can be performed using the update procedure as follows:
  \begin{itemize}
    \item Start with $G:= \emptyset$. We have $\#\alpha$ equals $0$ and $L_\alpha$ is empty.
    \item Add all vertices one by one and for each perform the update procedure.
    \item Add all edges one by one and for each perform the update procedure.
  \end{itemize}
  This gives a preprocessing procedure with the desired complexity. We now turn to the update procedure.

  Let $\C$ be a class of graph with low degree. Let $G=(V,E)$ be a graph in $\C$. Let $\e>0$, $r\in \NN$ and $\alpha(\x)$ be an $r$-local centered query and $k$ its arity. Let $\ez>0$ small  enough that: $\e \ge 6 \ez k^2 r \cdot |\alpha|$ and $\log(n) \le f(\ez)\cdot n^\ez$. (Recall $f(\ez)$ from the definition of classes of graphs with effectively low degree.)

  Let $G'=(V',E')$ the graph obtained from $G$ after the update.
  Let $n=|V'|$.
  We know that $d(G') \le f(\ez)\cdot n^\ez$. Let $L_u$ be the list of nodes $2kr$-impacted by the local update. By Lemma~\ref{lem-update}, we have that $|L_u|$ has size at most $O(d(G)^{2kr}) = O(n^{\ez \cdot 2kr})$ and can be computed in similar time. We then apply the following algorithm:

  For every $r$-impacted node $a_1\in L_u$, we study every possible tuple $\a = (a_1,\ldots,a_k) \subseteq N_{2kr}^{G'}(a_1)$.
  Then for every such tuple, we test by brut fore whether $\N_{2kr}^{G'}(a_1) \models \alpha(\a)$. As $\alpha$ is $r$-local centered, this is true if and only if  $G \models \alpha(\a)$. We also test whether $\a$ was already in $L_\alpha$. We finally update $L_\alpha$ and $\#\alpha$ accordingly.

  For the complexity analysis, the size of $\N_{2kr}^{G'}(a_1)$ is bounded by $d(G)^{2kr}$, i.e.~bounded by $n^{2\ez kr}$. Therefore the number of considered tuples is bounded by $(\N_{2kr}^{G'}(a_1))^k  =n^{2\ez k^2r}$.
  Then we need time $O(|\N_{2kr}^{G'}(a_1)|^{|\alpha|})$ to check with brute force whether $\N_{2kr}^{G'}(a_1) \models \alpha(\a)$, i.e.~$O(n^{2\ez kr \cdot |\alpha|})$.
  The procedures $add$ and $remove$ require time $O(\log(n))$ to keep $L_\alpha$ ordered, and thus $O(f(\ez)\cdot n^\ez)$. Therefore the time spend for every considered $\a$ is $O(n^{2\ez kr \cdot |\alpha|})$.
  Hence we spend time $O(n^{4 \ez k^2 r \cdot |\alpha|})$ for every element $a\in L_u$. At the end, the overall complexity is:\newline
   $O(|L_u|\cdot n^{4 \ez k^2 r \cdot |\alpha|}) = O(n^{6 \ez k^2 r \cdot |\alpha|}) = O(n^\e)$.

  This ends the proof of Lemma~\ref{lem-low-unary}
\end{proof}

We now have at our disposal a data structure that efficiently evaluates centered queries. In the upcoming sections we generalize it to deal with full \FO.

\subsection{The model checking problem}\label{subseq-low-mc}

This section proves Theorem~\ref{thm-low-mc}. The next lemma shows that Lemma~\ref{lem-low-unary} is enough to prove the result for sentences of the form~\ref{eq-nf-mc}. By Theorem~\ref{thm-nf-mc}, every first order sentence is a boolean combination of such sentences. As evaluating a boolean combination takes constant time (independent from the size of $G$), Theorem~\ref{thm-low-mc} follows.

\begin{lem}\label{lem-low-mc}
  For every class of graphs $\C$ with effectively low degree, for every $r$-local unary query $\alpha(x)$, for every integer $k$,  for every $G$ in $\C$ with $n:=|G|$, and for every $\e>0$, giving $L_\alpha$ and $\#\alpha$ as stated in Lemma~\ref{lem-low-unary}, we can test in time $O(n^\e)$ whether $G$ models
  $$\ph ~:= ~\exists x_1\ldots x_{k} ~\Big( \underset{1\le i < i' \le k}{\bigwedge}~\dist_{>2r}(x_i,x_{i'}) \wedge \underset{1\le i\le k}{\bigwedge} \alpha(x_i) \Big).$$
\end{lem}

\begin{proof}[Proof of Lemma~\ref{lem-low-mc}]
  Let $\C$ be a class of graphs with low degree, $G$ a graph of $\C$, $n:=|G|$, and $\e>0$.
  Let $\ez >0$ such that $\e \ge 3\ez r\cdot|\ph|$.

  We know that $d(G)\le f(\ez)\cdot n^\ez$. We then distinguish two cases whether $\#\alpha$ is smaller than $k\cdot d(G)^{2r}$.

  \begin{itemize}
    \item If $\#\alpha > k\cdot d(G)^{2r}$, there must be $k$ elements in $L_\alpha$ that are pairwise at distance greater than $2r$. In that case we know that $G\models\ph$.

    \item If $\#\alpha \le k\cdot d(G)^{2r}$, then $|N_r^G(L_\alpha)| \le k\cdot d(G)^{3r}$. Furthermore, it is always the case that $G \models \ph$ if and only if $\N_r^G(L_\alpha) \models \ph$. Therefore we can test whether $G\models \ph$ in time:
    $$ O\left( N_r^G(L_\alpha)^{|\ph|} \right)  \quad = \quad
    O\left( \big(k\cdot d(G)^{3r}\big)^{|\ph|} \right) \quad = \quad
    O\big(n^{3\ez r \cdot |\ph|}\big) \quad = \quad O\big(n^\e\big).$$\qedhere
  \end{itemize}
  This ends the proof of Lemma~\ref{lem-low-mc} and therefore of Theorem~\ref{thm-low-mc}
\end{proof}

\subsection{The enumeration problem}\label{subseq-low-enum}

The goal of this section is to prove the enumeration part of Theorem~\ref{thm-low-all}.
Thank to Corollary~\ref{cor-fo-queries}, every $\FO$ query is a boolean combination of sentences and $\FO$ queries of the form:
$\alpha_1(\x_1)\wedge \ldots \wedge \alpha_p(\x_p) \wedge \Delta_{2r,\tau}(\x).$
Furthermore, once the sentences have been replaced by their truth value, we are left with a strict disjunction of queries of the form above.

We can compute the data structure of Lemma~\ref{lem-low-mc} for each sentence of the decomposition. Unfortunately, the value of those sentences over the graph can be impacted by the update. Fortunately, there is only a constant number of possible valuations for the sentences.

Every possible evaluation of the sentences gives us a strict disjunction of queries of the form above. For each, we compute a data structure (described below) allowing us to enumerate their solutions with constant delay.
For the enumeration phase, we first test which sentences are true. This gives us which strict disjunction of queries of the form above needs to be enumerated. We then simultaneously enumerate all of their solutions in lexicographical order and with constant delay.

We therefore only prove the following lemma.
\begin{lem}\label{lem-low-enum}
  Let $\ph(\x)$ be an $\FO$ query of the form $\ph(\x) = \alpha_1(\x_1)\wedge \ldots \wedge \alpha_p(\x_p) \wedge \Delta_{2r,\tau}(\x)$,
  where $\alpha_i$ are $r$-local centered queries. Let $G$ be a graph from a class $\C$ with effectively low degree and $n:=|G|$. For every $\e>0$, there is a data-structure that:
  \begin{itemize}
    \item can be computed in time $O(n^{1+\e})$,
    \item can be updated in time $O(n^{\e})$ when $G$ is subject to local
    updates, and
    \item allows to enumerate $\ph(G)$ with constant delay and no additional preprocessing.
  \end{itemize}
\end{lem}

The proof goes by induction on $k$, which is the number of variables in $\ph$.

\begin{proof}[Proof of Lemma~\ref{lem-low-enum}]~
  Similarly to the proof of Lemma~\ref{lem-low-unary}, we only provide the update procedure as the preprocessing can be viewed as performing updates from the empty graph to the desired one.

  \medskip\noindent
  $\bullet$ If $k=1$. Then $p=1$ and this case has already been proved by Lemma~\ref{lem-low-unary}. Once $L_\ph$ has been stored on a RAM, we can simply run through it.

  \medskip\noindent
  $\bullet$ If $k>1$. Let $\z = (\x_1,\ldots,\x_{p-1})$. Consider the following
  two queries:
  $$ \ph_1(\z) := \exists \y_p ~ \ph(\z,\y_p) \quad \text{and} \quad
  \ph_2(\x_p)  := \alpha_p(\x_p).  $$

  It is easy to see that both $\ph_1$ and $\ph_2$ have strictly less than $k$
  variables. Therefore, the induction hypothesis holds for $\ph_1$ and $\ph_2$.
  A simple (and wrong) idea would be to enumerate all solutions $\c$ in $\ph_1(G)$ and for each of them enumerate all solutions $\a_p$ in $\ph_2(G)$ and output $(\c,\a_p)$ if and only if $G\models \dist_{>2r}(\c,\a_p)$. The issue is that given a tuple $\c$, we might encounter, in a row, a non constant amount of tuples $\a_p$ that do not satisfy this condition.

  We need, for any tuple $\c$ and any tuple $\a_p$ in $\ph_2(G)$, to be able to compute in
  constant time the smallest tuple $\a'_p$ in $\ph_2(G)$ that is bigger than $\a_p$ and that satisfies the condition: $G\models \dist_{>2r}(\c,\a'_p)$.

  An adequate technique to solve this issue has been used several times in previously published results. The firsts appearances where in~\cite[Section 3.5]{DBLP:conf/pods/DurandSS14}
  and~\cite[Section 4]{DBLP:conf/pods/KazanaS13}. Some variations have been used again in~\cite[Lemma 10.1]{DBLP:journals/tods/BerkholzKS18} and~\cite[Lemma 3.10]{DBLP:conf/pods/SchweikardtSV18}. We now state our version, 
  that is particularly suited to deal with tuples of elements.

  \begin{lem}\label{lem-skip}
    Let $\C$ be a class of graphs with low degree. Let $G$ be a graph in $\C$, $n:=|G|$, $\e>0$ and $L$ be a list of $k$-ary, $r$-centered tuples of elements of $G$ and let $k\in\NNpos$. There is a data structure computable in time $O(|L|\cdot n^{\e})$ and recomputable in time $O(n^\e)$ (when a tuple is added or removed from $L$) allowing us, given a tuple $\b$ in $L$ and a set $I$ of at most $k$ elements of $G$, to compute in constant time the next tuple $\b'$ in $L$ satisfying $\dist_{>r}(\b',a)$ for all $a$ in $I$.
  \end{lem}
  The idea is to compute a function $\sk(\b,I)$ that does what we want. The difficulty is that this function is way too big to be computed entirely. Fortunately only computing a small part of it (of pseudo linear size) is enough for our needs.
The proof of Lemma~\ref{lem-skip} is written in Appendix~\ref{app-skip}.
\end{proof}

\subsection{The testing problem}\label{subseq-low-test}
The goal of this section is to prove the testing part of Theorem~\ref{thm-low-all}.

\begin{proof}
  Let $G$ be a graph from a class $\C$ of low degree, let $n:=|G|$, and let $\ph(\x)$ be an $\FO$ query. Thanks to Corollary~\ref{cor-fo-queries}, we know that there is an $r$ 
  exponential in $|\ph|$ such that $\ph(\x)$ is equivalent to a boolean combination of $\FO$ sentences and $\FO$ queries of
  form~\ref{eq-nf-k}.

  Thanks to Theorem~\ref{thm-low-mc}, for every $\FO$ sentence of the decomposition, there is a data structure allowing us to test in constant time whether each sentence holds in the graph. Thanks to Lemma~\ref{lem-low-unary}, similar structures exist for any $r$-local centered queries $\alpha_i(\x_i)$ that appear in the decomposition. Moreover, the distance type query $\Delta_{2r,\tau}(\x)$ is a boolean combination of $r$-local centered queries. Hence we can also use the data structure provided by Lemma~\ref{lem-low-unary} for those queries.

  All of these structures can be computed in time $O(n^{1+\e})$ and updated in
  time $O(n^\e)$. Furthermore they are only a constant number of such queries,
  hence the overall complexity is exactly what we need.
  Finally, when a tuple is given, we can use those data structures to test in
  constant time which queries and sentences hold in the graph.
  Given that information, we can compute the boolean combination and answer the
  testing problem for $\ph(\x)$ over $G$ in constant time (depending only on $\C$ and $\ph$).
\end{proof}

\subsection{The counting problem}\label{subseq-low-count}
Here we prove the counting part of Theorem~\ref{thm-low-all}.

Similarly to the enumeration part, and thanks to Corollary~\ref{cor-fo-queries} we only provide an algorithm for queries of the form $\ph(\x) = \alpha_1(\x_1)\wedge \ldots \wedge \alpha_p(\x_p) \wedge \Delta_{2r,\tau} (\x)$.

The proof goes by induction on $p$, which is the number of connected components in the \linebreak distance-type~$\tau$. More precisely, we show that the problem of counting queries with several centered components can be reduced to the problem of counting local centered queries. Counting local centered queries is solved by Lemma~\ref{lem-low-unary}.

\begin{lem}\label{lemma-count-decomposition}
  Let $\ph(\x)$ be an $r$-local query of the form $\alpha_1(\x_1)\wedge \ldots\wedge \alpha_p(\x_p) \wedge \Delta_{2r,\tau} (\x)$ where all $\alpha_i$ are $r$-local centered queries.

  There are $r$-local centered queries $\ph_1,\ldots,\ph_m$ and an arithmetic combination $F$ such that for all graph $G$, we have $\#\ph(G)= F \big(\#\ph_1(G),\ldots,\#\ph_m(G)\big)$. Furthermore, $m = O(2^{pk})$.
\end{lem}

The proof of Lemma~\ref{lemma-count-decomposition} follows an inclusions-exclusion argument.

\begin{proof}[Proof of Lemma~\ref{lemma-count-decomposition}]
  Let $\ph(\x)$ be an $r$-local query of the form $\alpha_1(\x_1)\wedge \ldots\wedge \alpha_p(\x_p) \wedge \Delta_{2r,\tau}(\x)$.
  If $p=1$, we are done. If $p>1$. Let $\w = (\x_2,\ldots,\x_p)$. Consider the following three queries:
  \[\begin{array}{r l}
    \ph_1(\x_1)     &:= \alpha_1(\x_1), \\
    \ph_2 (\w)      &:=  \alpha_2(\x_2) \land \ldots \land \alpha_p(\x_p) \land \Delta_{2r,\tau'}(\w), \text{where } \tau' \text{ is the restriction of } \tau \text{ to } \w \\
    \ph_3(\x_1,\w)  &:=\ph_1(\x_1)\wedge \ph_2(\w)\wedge \dst{2r}(\x_1;\w).
  \end{array}\]

\noindent For all $(\a,\b)\in G^k$, we have that: $G\models \ph(\a\b)$\quad if and only if $G\models\ph_1(\a) \wedge \ph_2(\b) \wedge \dist_{>2r}(\a;\b)$,\\
hence \quad $\ph(G) = \ph_1(G)\bowtie \ph_2(G) \setminus \set{\a,\b \in G ~|~ \ph_1(\a)\wedge \ph_2(\b)\wedge \dst{2r}(\a,\b)}$,\\
which gives \quad $\ph(G) = \ph_1(G)\bowtie \ph_2(G) \setminus \ph_3(G)$.\\
And since $\ph_3(G) \subseteq \ph_1(G)\bowtie \ph_2(G)$, 
it follows that $\#\ph(G) = \# \ph_1(G) \cdot \# \ph_2(G) - \#\ph_3(G)$.

  It is easy to see that both $\ph_1$ and $\ph_2$ have strictly less than $p$ connected components in their distance type (respectively $1$ and $p-1$) and are $r$-locals. Therefore, the induction hypothesis holds for $\ph_1$ and $\ph_2$. I.e.~$\#\ph_1(G)$ and $\#\ph_2(G)$ are arithmetic combinations of number of solutions for $r$-local centered queries.

  We now turn our attention to $\ph_3$. We say that $(\x'_1;\ldots;\x'_{p'}) \ll (\x_1;\ldots;\x_p)$ if and only if:
  \begin{itemize}
  \item $(\x'_1;\ldots;\x'_{p'})$ is a partition of $\x$ with $p' < p$,
  \item $\x_1 \subsetneq \x'_1$,
  \item $\forall 1 < j \le p',$ there is a $i>1$ such that $\x'_j = \x_i$.
  \end{itemize}

  Basically, $\x'_1$ is the collapse of $\x_1$ and at least one of the $\x_i$. The other $\x_i$ remain unaltered. Given ${\cal I} :=(\x'_1;\ldots;\x'_{p'})$, for every $j\le p'$, we define:  $\alpha'_j(\x'_j)= \bigwedge\limits_{\setc{i \le p}{\x_i \subseteq \x'_j}} \alpha_i(\x_i)$.

  It follows from those definitions that:
  $$\ph_3(\x) = \bigvee\limits_{(\x'_1;\ldots;\x'_{p'})\ll (\x_1;\ldots;\x_{p})} \alpha'_1(\x_1) \wedge \ldots \wedge \alpha'_{p'}(\x'_{p'}) \wedge \Delta_{2r,\tau'}(\x).$$
  Moreover these disjunctions are strict, hence:
  $$\# \ph_3(G) = \sum\limits_{(\x'_1;\ldots;\x'_{p'})\ll (\x_1;\ldots;\x_{p})} \# \big( \alpha'_1(\x_1) \wedge \ldots \wedge \alpha'_{p'}(\x'_{p'}) \wedge \Delta_{2r,\tau'}(\x) \big).$$

  Since every query involved here has strictly less than $p$ connected components in its distance type, the induction hypothesis holds. There are at most $2^k$ queries involved. For each of them the induction hypothesis leads to $2^{k(p-1)}$ queries that are $r$-local centered such that computing the number of solutions for those queries give us the $\# \ph_3(G)$. Therefore the total number of queries built is $O(2^{pk})$.
\end{proof}

It is worth mentioning, that Lemma~\ref{lemma-count-decomposition} never mentions the degree of the structure and can be used in other contexts.

\section{Structure with bounded degree}\label{seq-bound}
This section we move to structures of bounded degree and our second contribution. 
\begin{thm}\label{thm-bound-mc}
  There is an algorithm that, upon input of integers $d,\ell$ and a graph $G$ of degree bounded by $d$ with $n:=|G|$, computes a data structure that:
  \begin{itemize}
    \item can be computed in time $n\cdot 2^{d^{2^{O(\ell)}}}$,
    \item can be recomputed in time $2^{d^{2^{O(\ell)}}}$ when $G$ is subject to a local update,
    \item enables us to test in time $2^{d^{2^{O(|\phi|+\ell)}}}$ whether $G\models \phi$ for any $\FO$ sentences $\phi$ of quantifier rank at most~$\ell$.
  \end{itemize}
\end{thm}

\begin{thm}\label{thm-bound-all}
  There is an algorithm that, upon input of integers $d,\ell$ and a graph $G$ of degree bounded by $d$ with $n:=|G|$, computes a data structure that:
  \begin{itemize}
    \item can be computed in time $n\cdot 2^{d^{2^{O(\ell)}}}$, and
    \item can be recomputed in time $ 2^{d^{2^{O(\ell)}}}$ when $G$ is subject to a local update.

    \item Furthermore, upon input of a $k$-ary query $\ph(\x)$ of quantifier rank $q$ with $k+q\le \ell $, after a further preprocessing in time $2^{d^{2^{O(|\phi|+\ell)}}}$, the data structure enables us:
    \begin{enumerate}
      \item~{\textit{\small (enum:)}} to enumerate with delay $O(k^3)$ the set $\ph(G)$,
      \item~{\textit{\small (test:)}} upon input of a $k$-tuple $\a$, to test in time $O(k^2)$ whether $G\models \ph(\a)$, and
      \item~{\textit{\small (count:)}} to compute in time $O(1)$ the number $|\ph(G)|$ of solutions for $\ph$ over $G$.
    \end{enumerate}
  \end{itemize}
\end{thm}

\nc{\iso}{\simeq}

\subsection{Basic data structure}\label{subseq-bound-basic}
Similarly to Section~\ref{seq-low-degree} we start with centered queries. We describe what we call the basic data structure ($BDS$).
Informally, the basic data structure contains all possible graphs of bounded degree and bounded radius and some mappings 
between all these small graphs and our larger structure.

More formally, the data structure is defined for two integers $d,\ell$, and a graph $G=(V,E)$ of degree at most $d$.
The basic data structure is composed of the following elements:
\begin{itemize}
  \item The set $\H$ of every (up to isomorphism) graph $H$ of degree bounded by $d$, one marked node $c^H$,
  and radius at most $r\ell$ around $c^H$,
  where $r=2^{O(\ell)}$.
  \item A function $\X(\cdot)$ that associates each $a\in V$ to the only $H\in\H$ such that there is an isomorphism from $\N_{r\ell}^G(a)$ to $H$ that maps $a$ to $c^{H}$.
  \item For all $H\in\H$, the list $L_{H}$ of all $a\in V$ such that 
  $\X(a)=H$.
  \item For all $H\in\H$, the number $\#L_{H}$ that is the size of $L_H$.
  \item For every $a$ in $V$, we fix an isomorphism $f_a(\cdot)$ from $\N_{r\ell}^G(a)$ to $H$ (that maps $a$ to $c^{H}$).
\end{itemize}
This ends the description of the basic data structure $BDS(d,\ell,G)$.

Note that an element $a$ in $G$ appears in only one of the $L_H$.

\begin{lem}\label{lem-bound-BDS}
  Let $d,\ell$ in $\NNpos$ and a graph $G$ of degree bounded by $d$. We can:
    \begin{itemize}
    \item compute $BDS(d,\ell,G)$ in time $O(|G|)$, and
    \item update $BDS(d,\ell,G)$ into $BDS(d,\ell,G')$ when a local update changes $G$ into $G'$ in time $O(1)$.
  \end{itemize}
\end{lem}

Here again, the $O(\cdot)$ hide constant factors that depend on $d$ and $\ell$.
Similarly to other proofs of this paper, we do not prove the first item of Lemma~\ref{lem-bound-BDS} as it is enough to apply the update procedure from the empty graph to the desired one by adding all vertices and all edges one by one.

\begin{proof}[Proof of Item 2 of Lemma~\ref{lem-bound-BDS}]
  Let $d,\ell,G$ and $BDS(d,\ell,G)$. Let $G'$ be the update of $G$, let $r=2^{O(\ell)}$ given by Corollary~\ref{cor-fo-queries} and let $L_u$ be the list of $r\ell$-impacted nodes. 

  For all $a$ in $L_u$, we perform the following procedure:
  \begin{itemize}
    \item Let $H= \X(a)$.
    \item Remove $a$ from $L_H$ and put $\#L_H := \#L_H -1$.
    \item Compute $\N_{r\ell}^{G'}(a)$ by performing a simple breadth-first search in $G'$ starting from $a$.
    \item Scan through $\H$ to find $H'$ such that $\N_{r\ell}^{G'}(a)\iso H'$ (with $a$ mapped to $c^{H'}$).
    \item Set $\X(a) := H'$.
    \item Compute an isomorphism $f_a(\cdot)$ from $\N_{r\ell}^{G'}(a)$ to $H'$ (with $a$ mapped to $c^{H'}$).
    \item Set $\#L_{H'} \gets \#L_{H'}+1$ and add $a$ into $L_{H'}$.
  \end{itemize}

  This computes what we need. Let us now look closely at the complexity of this algorithm. By Lemma~\ref{lem-update} $L_u$ has size $O(d^{\,r\ell})$ and can be computed in time $O(d^{\,r\ell})$. For each $a$ in $L_u$, step (1), (2), (5) and (7) can be performed in time $O(1)$ with our computational model. Step (3) can be performed in time $O(d^{r\ell})$. Step (6) can be performed in time $O(2^{d^{r\ell}})=2^{d^{2^{O(\ell)}}}$. For step (4), the number of elements in $\H$ is bounded by $O(2^{(d^{r\ell})^2})$. For each we perform a test similar to the one in Step (6), leading again to a total complexity of $2^{d^{2^{O(\ell)}}}$.
\end{proof}
We now turn to the application of the Basic Data structure.

\subsection{The model checking problem}\label{subseq-bound-mc}

In this section we prove Theorem~\ref{thm-bound-mc}. The data structure is almost the one defined in Section~\ref{subseq-bound-basic}.

Let $d,\ell$ be two integers, and $G$ a graph of degree $d$. Assume that $BDS(G,d,\ell)$ has been computed. We are now given an $FO$ sentence $\phi$ of quantifier rank at most $\ell$. By Theorem~\ref{thm-nf-mc}, $\phi$ is equivalent to a boolean combination of sentences of the form:
$$\exists x_1\ldots x_s ~\Big( \underset{1\le i \le j \le s}{\bigwedge}~\dist_{>2r}(x_i,x_j) \wedge \underset{1\le i\le s}{\bigwedge} \alpha(x_i) \Big)$$
where $\alpha$ is an $r$ local query and $r= 2^{O(\ell)}$. We show how to test whether $G\models\ph$ for sentences of the form above.

\begin{itemize}
  \item The first thing we do is to compute $\#\alpha(G)$ that is the number of elements of $G$ satisfying $\alpha$. Note that $G\models \alpha(a)$ if and only if $\N_{r\ell}^G(a) \models\alpha(a)$ and therefore if and only if $H\models \alpha(c^H)$, where $H=\X(a)$ i.e.~the graph isomorphic to $\N_{r\ell}^G(a)$.

  Hence we have that $\#\alpha(G) = \sum\limits_{\setc{H \in \H}{H\models \alpha(c^H)}} \#L_H$. Thus we test for every $H$ in $\H$ whether $H\models \alpha(c^H)$, and if so, we add $\#L_H$ to the total.

  \item Once $\#\alpha(G)$ has been computed we consider two cases:
  \begin{itemize}
    \item If $\#\alpha(G)> k\cdot d^{2r}$, there must be $k$ elements satisfying $\alpha(x)$ at distance greater than $2r$ from each other. We can therefore assert that $G\models \ph$.
    \item If $\#\alpha(G) \le k\cdot d^{2r}$, we compute the list $\alpha(G)$ of all elements satisfying $\alpha(x)$. Similarly to the first phase, we have that $\alpha(G) = \bigcup\limits_{\setc{H \in \H}{H\models \alpha(c^H)}} L_H$.

    We then define $\bm\tilde{G}:= \bigcup\limits_{a \in \alpha(G)} N_r^G(a)$. We have that $G \models \ph$ if and only if $\bm\tilde{G}\models \ph$. We then use any naive algorithm to test whether $\bm\tilde{G}\models \ph$.
  \end{itemize}
\end{itemize}

This completes the answering phase. Let us now look at the complexity of this procedure. The number of graphs in $\H$ and their sizes is a function of $d$ and $\ell$ and is independent of $|G|$. Therefore computing $\#\alpha(G)$ only took time $O(1)$.

More precisely, there are $O\left(2^{(d^{r\ell})^2}\right)$ elements in $\H$ each of size $O(d^{r\ell})$.
Testing whether a given $H$ models $\alpha(c^H)$ can therefore be performed in time $O\left(2^{d^{r\ell}}\right)$. The overall complexity of the first phase is therefore $2^{d^{2^{O(\ell)}}}$.

During the second phase, the first case is just a test that can be performed in time $O(1)$. In the second case, we have that $|\bm\tilde{G}|\le \#\alpha(G)\cdot d^r$. Since $\#\alpha(G) \le k\cdot d^{2r}$, we have that $|\bm\tilde{G}|\le k\cdot d^{3r}$. We then check whether $\bm\tilde{G}\models \ph$ in time $O\big(|\bm\tilde{G}|^{|\ph|}\big)$.

Hence the total time of the second phase $O\big( (k\cdot d^{3r})^{|\phi|} \big) = 2^{d^{2^{O(|\phi|+\ell)}}}$

\subsection{The enumeration problem}\label{subseq-bound-enum}

This section is devoted to the enumerating part of Theorem~\ref{thm-bound-all}. First we show that the Basic Data Structure is enough if we allow the delay to depend on both the degree $d$ and the radius $r$. In Appendix~\ref{app-fast-enum}, we show that similarly to~\cite[Theorem 9.2]{DBLP:journals/tods/BerkholzKS18} we can make the delay be $O(k^3)$.

The first algorithm only uses the Basic Data Structure and works by induction on the number $k$ of free variables in $\ph$. The global proof is a combination of these two lemmas.

\begin{lem}\label{lem-bound-enum-one}
   For every fixed $G,d,\ell$, once $BDS(G,d,\ell)$ is computed, then for any unary \FO query $\ph(x)$ of quantifier rank $q$ such that $q < \ell$, without additional preprocessing, we can enumerate the set $\ph(G)$ with constant delay (depending on $d$ and $\ell$).
\end{lem}

\begin{lem}\label{lem-bound-enum-last}
   For every fixed $G,d,\ell$, once $BDS(G,d,\ell)$ is computed, for every $k$-ary \FO query $\ph(\x)$ and quantifier rank $q$ such that $k+q\le \ell$, without additional preprocessing and for every $k-1$ tuple $(a_1,\ldots a_{k-1})$, we can enumerate the set of all $b$ in $G$ such that $G\models \ph(a_1,\ldots,a_{k-1},b)$ with constant delay (depending only on $d$ and $\ell$).
\end{lem}

The following induction on $k$ explains why these two lemmas are enough. Fix $G,d,\ell$ and assume that $BDS(G,d,\ell)$ has been computed. We are now given an $\FO$ query $\ph(\x)$ of arity $k$ and quantifier rank $q$ such that $k+q\le \ell$. If $k=1$, Lemma~\ref{lem-bound-enum-one} explains how to enumerate the solutions. If $k>1$, we define $\ph'(x_1,\ldots,x_{k-1}) := \exists y \ph(x_1,\ldots,x_{k-1},y)$. Note that $\ph'$ has quantifier rank $q+1$ but arity $k-1$. Therefore by induction we can enumerate without further preprocessing the set $\ph'(G)$. For each tuple $(a_1,\ldots,a_{k-1})$ produced by this procedure, we apply Lemma~\ref{lem-bound-enum-last} to complete this tuple into a solution for $\ph$.

It is critical to see that every tuple in $\ph'(G)$ can be completed into a solution for $\ph$. Otherwise, we might encounter a lot of consecutive solutions for which the algorithm of Lemma~\ref{lem-bound-enum-last} would produce the empty set, making the delay not constant.

It only remains to prove the two lemmas. Here we only provide a sketch of those proofs, the complete ones are in the appendix, Section~\ref{app-simple-enum}.

\noindent{\it Proof sketch of Lemma~\ref{lem-bound-enum-one}.}
  Once $BDS(G,d,\ell)$ has been computed, the list $\ph(G)$ for a unary query $\ph$ is just a concatenation of lists $L_H$ from the Basic Data Structure.

\noindent{\it Proof sketch of Lemma~\ref{lem-bound-enum-last}.}
  The algorithm distinguishes two cases. If $b$ must be at distance less than $2r$ to one of the $a\in\a$, there is only a constant number of possible $b$. If $b$ must be at distance greater than $2r$ from $\a$, then all $b$ that satisfies some local condition (not depending on $\a$) are solution except those that are close to $\a$. Since they are only a constant number of $b$ close to $\a$, even if we encounter all of them in a row, the delay remains bounded by a constant (depending on $d$ and $r$).

\subsection{The testing problem}\label{subseq-bound-test}
In this section we prove the testing part of Theorem~\ref{thm-bound-all}. Similarly to Section~\ref{subseq-bound-mc}, we show that the basic data structure is enough for completing this task.

\paragraph*{Naive algorithm}
Let $\ph(\x)$ be an $FO$ query of arity $k$ and quantifier rank $q$ with $k+q\le \ell$.

By Corollary~\ref{cor-fo-queries}, $\ph$ is equivalent to a boolean combination of sentences of the form of Equation~\ref{eq-nf-mc} and $r$-local queries where $r=2^{O(\ell)}$.

Since the sentences can be evaluated using Theorem~\ref{thm-bound-mc}, we are only left with $r$-local queries. It turns out that over structures with bounded degree, local queries can always be tested in constant time because for every $r$-local queries $\ph(\x)$ and every tuple $\a$ in $G$:
$$ G \models \ph(\a) \quad \text{ if and only if } \quad \N_r^G(\a) \models \ph(\a)$$

We can compute $\N_r^G(\a)$ in constant time and then test whether any $r$-local query holds in constant time (with constant factors depending on $r$ and $\ell$).

\paragraph*{Faster algorithm}
In~\cite[Theorem 6.1]{DBLP:journals/tods/BerkholzKS18} it is shown, that the testing procedure can be performed in time $O(k^2)$ where $k$ is the arity of the query. One can split the answering phase in two, the query is given first and the tuple latter. We can enrich the data structure to match the complexity of~\cite{DBLP:journals/tods/BerkholzKS18}.

\begin{lem}\label{fast-test}
  For every fixed integers $d$, $\ell$ in $\NNpos$ there is an algorithm that,
  upon input of a graph $G$ of degree bounded by $d$, computes a data structure that:
  \begin{itemize}
    \item can be computed in time $|G|\cdot 2^{d^{2^{O(\ell)}}}$, and
    \item can be recomputed in time $2^{d^{2^{O(\ell)}}}$ when $G$ is subject to a local update.

    Furthermore, upon input of a $k$-ary query $\ph(\x)$ of quantifier rank $q$ with $k+q\le \ell$, the data structure allows to perform a final task in time triply exponential in $\ph$ such that afterward upon input of a $k$-tuple $\a$, to test in time $O(k^2)$ whether $G\models \ph(\a)$
  \end{itemize}
\end{lem}

The proof of Lemma~\ref{fast-test} is in Appendix~\ref{app-fast-test}.

\subsection{The counting problem}\label{subseq-bound-count}
Here we prove the counting part of Theorem~\ref{thm-bound-all}. Again we show that the basic data structure is enough for completing this task.
Let $(G,d,\ell)$ fixed. Let $r = 2^{O(\ell)}$, and assume that $BDS(G,d,\ell)$ has been computed.

Let $\ph(\x)$ be an $FO$ query of arity $k$ and quantifier rank $q$ with $k+q\le \ell$.

By Corollary~\ref{cor-fo-queries}, $\ph$ is equivalent to a boolean combination of sentences of the form of Equation~\ref{eq-nf-mc} and $r$-local queries.
Furthermore, once all sentences have been evaluated, we are left with a strict disjunction of queries of the form:
$ \alpha_1(\x_1)\wedge \ldots\wedge \alpha_p(\x_p) \wedge \Delta_{2r,\tau}(\x)$ where $\alpha_i(\x_i)$ are $r$-local centered queries.
We therefore only prove how to count the solutions for a query of the form above, the total number can be obtained by summing the different results.

In the upcoming Lemma~\ref{a}, we show that the Basic Data Structure is enough when $p=1$. The general case is then solved by Lemma~\ref{lemma-count-decomposition} that state that when $p > 1$ the number of solutions is an arithmetical combination of the number of solutions for at most $2^{pk}$ new $r_q$-local centered queries (centered meaning that $p=1$).

\begin{lem}~\label{a}
  For every fixed $(G,d,\ell)$, once $BDS(G,d,\ell)$ is computed. For any $r$-local centered query $\alpha(\x)$ with $k$ free variables and quantifier rank $q$ (with $k+q\le \ell$), the number $\#\ph(G)$ can be computed in time $2^{d^{2^{O(\ell)}}}$.
\end{lem}

\begin{proof}
  Let $x_1$ be the first element of $\x$. For all tuple $\a$, we have that $G \models \alpha(\a)$ if and only if $H \models \alpha(f_{a_1}(\a))$. Where $H=\X(a_1)$ is the graph isomorphic to $\N^G_{r_\ell}(a_1)$ and $f_{a_1}(\cdot)$ is the isomorphism from $\N^G_r(a_1)$ to $H$.
  Remember that given a graph $H$ in $\H$, we have access to the list $L_H$ of all $a$ in $G$ such that $\X(a)=H$ and to $\#L_H$, the size of this list.

  For all $H$ in $\H$ we compute $\#\alpha(H,c^H)$ that is the number of $k$-tuple $\a$ in $H$ with $a_1=c^H$ such that $H\models \alpha(\a)$. We obtain that:
  $\#\alpha(G) = \sum\limits_{H\in\H} \big( \#\alpha(H,c^H) \cdot \#L_H \big).$

  This ends the proof of Lemma~\ref{a}. The time needed for this final phase is triply exponential in $\ell$ because there are only $2^{d^{O(r\ell)}}$ graphs in $\H$, and for a fixed $H$, the number $\#\ph(H,c^H)$ can be computed in time $2^{d^{O(r\ell)}}$.
\end{proof}

\section{Conclusion}
We proved that over structures with bounded or low degree and for every first-order query, we can efficiently compute and recompute a data structure allowing to answer in constant time (or with constant delay) the model-checking, testing, counting and enumeration problems. It was already known for classes of graphs with bounded degree~\cite{DBLP:journals/tods/BerkholzKS18} but not for classes of graph with low degree.

In the case of structure with bounded degree we showed that, without increasing the preprocessing time,	we can compute a data structure that does not depend on the query but only on its quantifier rank. We can therefore compute a single preprocessing that can later be used for a variety of queries.

\bibliographystyle{plain}
\bibliography{biblio}

\begin{appendix}
  \section{Appendix for Section~\ref{seq-low-degree}}

  \subsection[The enumeration problem]{Complement for the enumeration: Skip elements of a list}\label{app-skip}
  In this section we prove Lemma~\ref{lem-skip}, reformulated as follow:

  \begin{lem}[Skip pointers]\label{lemma-skip-pointer}
    For every graph $G$ with $n:=|G|$, every $r,k$ and $m$ in $\NN$, and every list of $r$-centered tuples $L\subseteq V^m$, there is a data structure, computable in time $O\big(|L|\cdot d(G)^{kr}\big)$ and recomputable in time $O(d(G)^{(krm)^2})$ allowing us, when given a tuple $\b$ in $L$ and a set $I$ of at most $k$ elements of $G$, to compute in constant time the tuple $\sk(\b,I)$ that is the next tuple $\b'$ of $L$ that satisfies $\b' \cap N^G_r(I) = \emptyset$.
  \end{lem}

  Note that for classes of graphs with low degree, the complexity above become $O(|L|\cdot n^\e)$ for the preprocessing and $O(n^\e)$ for the update. Similarly, for classes graphs with bounded degree, it becomes $O(|L|)$ and $O(1)$ (with constants factors at most doubly exponential in $r$).
  \begin{proof}
    From now on we fix $r,k,m$, $G$, and $L$ as in the statement of the lemma and let $d$ be the degree of $G$.

    The domain of the $\sk(\cdot,\cdot)$-function is too big so we cannot compute it during the preprocessing phase. Fortunately, computing only a small part of it will be good enough for our needs. For each tuple $\b$ of $L$, we define by induction a set $\SC(\b)$ of sets of at most $k$ elements. We start with $\SC(\b)=\emptyset$ and then proceed as follows.
    \begin{itemize}
      \item For all tuple $\b$ of $L$ and for all element $a$ in $G$ with $\dst{r}(a;\b)$, we add $\set{a}$ to $\SC(\b)$.
      \item For all nodes $\b$ of $L$, for all sets $I$ of elements from $G$, and all element $a$ of $G$, if $|I|<k$ and $I \in \SC(\b)$ and $\sk(\b,I) \cap  N^G_r(a) \neq \emptyset$, then we add $\set{I \cup \set{a}}$ to $\SC(\b)$.
    \end{itemize}

    In the preprocessing phase we will compute $\sk(\b,I)$ for all tuples $\b$ of $L$ and all sets $I\in\SC(\b)$. Before explaining how this can be accomplished within the desired time constraints, we first show that this is sufficient for deriving $\sk(\b,I)$ in constant time for	all tuples $\b$ and all sets $I$ consisting of at most $k$ elements of $G$.

    \begin{cla}\label{NE1}
      Given a tuple $\b$ of $L$, a set $I$ of at most $k$ elements of $G$, and $\sk(\c,J)$ for all tuples $\c$ that are after $\b$ in $L$ and all sets $J\in \SC(\c)$, we can compute $\sk(\b,I)$ in constant time.
    \end{cla}

    \begin{proof}
      If $\b$ is the last tuple of $L$ then either $\sk(\b,I)=\b$ or $\sk(\b,I)=\nul$. Otherwise let $\b'$ be the tuple following $\b$ in $L$. We consider two cases (testing in which case we fall can be done in constant time).\\
      \noindent\emph{Case~$1$:} $\b'\cap N^G_r(I)=\emptyset$. In this case, $\b'$ is $\sk(\b,I)$ and we are done.

      \noindent\emph{Case~$2$:} $\b'\cap N^G_r(I)\neq\emptyset$. So there is an $a\in I$ such that $\b'\cap N^G_r(a)\neq\emptyset$ Let $J$ be a maximal (w.r.t.\ inclusion) subset of $I$ in $\SC(\b')$. Since $\set{a} \in \SC(\b')$, we know that $J$ is non-empty.

   		We claim that $\sk(\b',J)=\sk(\b,I)$. To prove this, let us first assume for contradiction that $\sk(\b',J) \cap N_r^G(a)\neq \emptyset$ for some $a	\in I$. By definition, this implies that $a$ is not in $J$. Hence $|J|<|I| \le k$. Thus, by definition of $\SC(\b')$ we have $J\cup\set{a} \in \SC(\b')$ and $J$ was not maximal.

      Moreover, by definition of $\sk(\b',J)$, every tuples in $L$ between $\b'$ and $\sk(\b',J)$ intersect $N^G_r(a')$ for some $a'$ in $J$ (and therefore $a'$ in $I$). The claim follows.
    \end{proof}

    We now show that $\SC(\b)$ is small for all tuple $\b$ of $L$ and that we can compute efficiently $\sk(\b,I)$ for all tuple $\b$ and all sets $I\in \SC(\b)$.
    \begin{cla}\label{NE2}
      For each tuple $\b$ of $L$, $|\SC(\b)|$ has size $O\big((m\cdot d(G)^r)^k\big)$. Moreover, it is possible to compute $\sk(\b,I)$ for all tuples $\b$ of $L$ and all sets $I\in \SC(\b)$ in time $O\big(|L|\cdot (m\cdot d(G)^r)^k\big)$.
    \end{cla}
    Remember that $m$ is the length of the tuples that are in $L$.

    \begin{proof}
      We start by proving the first statement, and afterwards we use Claim~\ref{NE1} to show that we can compute these pointers inductively.

      By $\SC_p(\b)$ we denote the subset of $\SC(\b)$ of sets $I$ with $|I| \le p$. We know that $|\SC_1(\b)| \leq m\cdot d(G)^r$ for all tuples $\b$ of $L$. For the same reason, we have that $|\SC_{p+1}(\b)|$ is of size at most $O(m\cdot d(G)^r\cdot |\SC_p(\b)|)$. Therefore, for all $\b\in L$, we have:
      $$|\SC(\b)| \ \ = \ \ |\SC_k(\b)| \ \ \leq \ \ O\big((m\cdot d(G)^r)^k\big).$$

      We compute the pointers for $\b$ from $\b_\textit{last}$ to $\b_\textit{first}$, where $\b_\textit{last}$ and $\b_{\textit{first}}$ are respectively the last and first tuples of $L$. Given a tuple $\b$ in $L$, assume we have computed $\sk(\c,J)$ for all $\c$ after $\b$ in $L$ for all $J\in \SC(\c)$. We then compute $\sk(\b,I)$ for $I \in \SC(\b)$ using Claim~\ref{NE1}.

      At each step the pointer is computed in constant time. Since there are $O\big(|L|(kd(G))^{kr}\big)$ of them, the time required to compute them is as desired.
    \end{proof}

    So far we have not use the fact that $L$ only contains $r$-centered tuples. It is indeed only use for the update procedure. The key observation is that $\sk(\b,I)$ cannot be arbitrary far from $\b$ in the list.
    \begin{cla}
      When a tuple $\b$ is added to or removed from the list $L$, for every tuple $\c$ in $L$ that is not among the $\big(k\cdot d(G)^{3mr}\big)^m$ tuples that directly precede (or was preceding) $\b$, we have that: $\SC(\c)$ and $\sk(\c,I)$ for every $I\in\SC(\c)$ is unaffected by the addition or deletion of $\b$.
    \end{cla}

    If a tuples or arity $m$ is $r$-centered then it is included in the $2mr$ neighborhood of any of it's element.
    Therefore, there are at most $\Delta = \big(kd(G)^{3mr}\big)^m$ possible $r$-centered tuples of arity $m$ that intersect the $r$ neighborhood of $I$.

    When we remove or add a tuples $\b$ in $L$, we only have to recompute the skip pointers of the $\Delta$ elements that are just before $\b$ in $L$. Which can be done in total time $O(d(G)^{(krm)^2})$.

    The combination of these three claims proves Lemma~\ref{lemma-skip-pointer}.
  \end{proof}

  \section{Appendix for Section~\ref{seq-bound}}
  \subsection{Simple enumeration algorithm}\label{app-simple-enum}
  In this section we formally prove Lemmas~\ref{lem-bound-enum-one} and~\ref{lem-bound-enum-last}.
  \begin{proof}[Proof of Lemma~\ref{lem-bound-enum-one}]
    Let $(G,d,\ell)$ fixed and assume that $BDS(G,d,\ell)$ has been computed. Let $\ph(x)$ be a unary $\FO$ query of quantifier rank $q$ such that $q < \ell$. By Corollary~\ref{cor-fo-queries}, $\ph$ is equivalent to a boolean combination of sentences of the form of Equation~\ref{eq-nf-mc} (for $r=r_q$) and $r_q$-local queries. We then do the following:
    \begin{itemize}
      \item First, for every sentence of the decomposition, test their truth value on $G$.
      \item Let $\alpha(x)$ be the boolean combination of the remaining $r_q$-local queries.
      \item For an element $a$ we have that $G\models \alpha(a)$ if and only if $\X(a) \models \alpha(c^H)$.\\
      Therefore we have that $\setc{a\in G}{G\models \alpha(a)} =  \bigcup\limits_{\setc{H\in\H}{H\models \alpha(c^H)}} L_H$.
      \item We select then all $H$ in $\H$ such that $H\models \alpha(c^H)$ and $\#L_H >0$.
      \item Finally, we run through the lists $L_H$ of all selected $H$ from the previous step.
    \end{itemize}
    This ends the proof of Lemma~\ref{lem-bound-enum-one}. For the complexity analysis, the firsts four items require a total time triply exponential in $|\ph|$. Once this is done the delay is truly $O(1)$ as we just read lists.
  \end{proof}

  \begin{proof}[Proof of Lemma~\ref{lem-bound-enum-last}]
    Let $(G,d,\ell)$ fixed and assume that $BDS(G,d,\ell)$ has been computed. Let $\ph(\x,y)$ be an $\FO$ query of quantifier rank $q$ and arity $k$ such that $k+q\le \ell$. By Corollary~\ref{cor-fo-queries}, $\ph$ is equivalent to a boolean combination of sentences of the form of Equation~\ref{eq-nf-mc} (for $r=r_q$) and $r_q$-local queries. We are then given a tuple $\a=(a_1,\ldots,a_{k-1})$ and we want to enumerate all $b$ such that $(\a,b)\in \ph(G)$.
    \begin{itemize}
      \item First, for every sentence of the decomposition, test their truth value on $G$.
      \item We are left with a strict disjunction of queries of the form \\$\alpha_1(\x_1)\wedge \ldots \wedge \alpha_p(\x_p,y) \wedge \tau_{2r_q} (\x_1;\ldots;\x_p)$ where all $\alpha_i(\x_i)$ are $r_q$-locals (we can assume without loss of generality that $y$ appears on $\alpha_p$). We enumerate the appropriate $b$ for each query from the disjunction one by one.
      \item We now have a tuple $\a= \a_1,\ldots,\a_p$ and a query $\alpha_1(\x_1)\wedge \ldots \wedge \alpha_p(\x_p,y) \wedge \tau_{2r_q} (\x_1;\ldots;(\x_p,y))$. We check whether $G\models\tau_{2r_q} (\x_1;\ldots;\x_{p-1})$ and $G\models\alpha_i(\a_i)$ for all $i<p$. If this is not the case, we can safely enumerate the empty set.
      \item It remains to enumerate all $b$ such that $G\models \alpha_p(\a_p,b)$ and $\dist(a_i,b)>2r_q$ for all $a_i \not\in \a_p$. We distinguish two cases whether $\a_p=\emptyset$.
      \begin{itemize}
        \item Case $1$, $\a_p\neq\emptyset$. Let $a_j$ be an element of $\a_p$. All matching $b$ must be in $\N^G_{kr_q}(a_j)$ and therefore $G\models (\alpha_p(\a_p,b))$ if and only if $\N^G_{r_\ell}(a_j) \models (\alpha_p(\a_p,b))$.
        We can therefore look at all the elements of $N^G_{2kr_q}(a_j)$ (there are only $d^{2kr_q}$), and for each test whether $G\models \alpha_p(\a_p,b)$ and $\dist(a_i,b)>2r_q$ for all $a_i \not\in \a_p$.

        We can produce the list of these elements in time $O(2^{d^{r_\ell}})$ and then enumerate the list.

        \item Case $2$, $\a_p=\emptyset$. In this case, we run through the list of all $b$ such that $G\models\alpha_p(b)$. For each such $b$, we test in constant time whether it is at distance greater than $2r_q$ to $\a$. If so we output it, otherwise we go to the next element of the list. We now explain how to run through all $b$ such that $G\models\alpha_p(b)$ and why this procedure as constant delay.

        To compute the list of all $b$ such that $G\models\alpha_p(b)$, it is enough to see that $G\models\alpha_p(b)$ if and only if $H\models \alpha_p(c^H)$ where $H= \X(a)$. Therefore we have that: $$\setc{b\in G}{G\models\alpha_p(b)} ~=~ \bigcup\limits_{\setc{H\in\H}{H\models\alpha_p(c^H)}}L_H$$
        We can test in time in constant time whether a graph $H$ satisfies the condition and there is only a constant number of them. Hence in constant time (triply exponential in $\ell$), we can run through all $b$ such that $G\models\alpha_p(b)$.

         For the constant delay, it is enough to see that testing whether $b$ is at distance greater than $2r_q$ from $\a$ can be performed in constant time, more precisely $O(kd^{2r_q})$. Furthermore at most $k\cdot d^{2r_q}$ element are at distance less than $2r_q$ to $\a$ and will not be enumerated. This makes the delay constant, i.e.~only triply exponential in $\ell$.
      \end{itemize}
    \end{itemize}
    This ends the proof of Lemma~\ref{lem-bound-enum-last}.
  \end{proof}

  \subsection{Fast enumeration algorithm}\label{app-fast-enum}
  In this section we explain how to improve the data structure for the enumeration part of Theorem~\ref{thm-bound-all}. Ideally, we would like the delay to be independent from the degree of the graph.

  \begin{thm}\label{fast-enum}
      Let $d,\ell$ in $\NNpos$ and a graph $G$ of degree bounded by $d$. There is a data structure that:
      \begin{itemize}
        \item can be computed in time $O(|G|)$,
        \item can be recomputed in constant time (triply exponential in $\ell$) when $G$ is subject to a local update, and
        \item upon a query $\ph(\x)$ of arity $k$ and quantifier rank $q$ with $k+q\le \ell$ perform a final task in constant time (triply exponential in $|\phi|+\ell$) such that afterward, the solution for $\ph(\x)$ over $G$ can be enumerated with delay $O(k^3)$.
      \end{itemize}
  \end{thm}

  The data structure for Theorem~\ref{fast-enum} combine ideas from different parts of previous algorithms.
  We first use the efficient testing procedure i.e.~Lemma~\ref{fast-test}, precisely the fact that there is a data structure allowing to test the distance up to some threshold in time $O(k^2)$.
  We also use the Lemma~\ref{lem-skip} that allows to skip some part of a list.
  This enable us to improve the two lemmas (Lemmas~\ref{lem-bound-enum-one} and~\ref{lem-bound-enum-last}) that are at the core of the proof of the enumerating part of Theorem~\ref{thm-bound-all}.

  The proof of Lemma~\ref{lem-bound-enum-one} already matches our new standards, i.e.~the list of solution for a unary query can be enumerated with delay truly $O(1)$, after an additional preprocessing depending on the degree of the graphs (triply exponential in the size of the query).

  Concerning the improvement of Lemma~\ref{lem-bound-enum-last}, we need to improve the structure using Lemma~\ref{lem-skip} and Lemma~\ref{fast-test}. For each list $L_H$ of the Basic Data Structure, we compute the skip pointers $\sk(b,I)$ (as in Lemma~\ref{lem-skip}). Additionally, for every distance query $\dst{r}(x,y)$ with $r \le 2^\ell$, we maintain the data structure allowing to test the solutions in time $O(1)$.

  During the enumeration phase, instead or running a BFS algorithm, we can test all the distances in time $O(k^2)$. We can therefore test whether we are looking at a solution in time $O(k^2)$. If not we can use a skip pointer that directly gives us the next solution.

  The procedure in the proof of Lemma~\ref{lem-bound-enum-last} work in time $O(k^2)$ after the additional preprocessing (triply exponential in the size of the query).

  The overall delay becomes $O(k^3)$ because of the $k$ recursive calls produced by the global procedure.

  \subsection{Fast algorithm for the testing problem}\label{app-fast-test}
  \begin{proof}[Proof of Lemma~\ref{fast-test}]
  The only thing that we add (before knowing the query) is a structure that contains the value of the function $\delta_{r}(x,y)$ that test the distance up to threshold $r$. I.e.~$\delta_{r}(a,b)=r'$ with $r'\le r$ if $G\models \dist_{=r'}(a,b)$ and $\delta_{r}(a,b)$ is undefined otherwise. Here, $r$ is obtain from $\ell$ with Corollary~\ref{cor-fo-queries}, and $r=2^{O(\ell)}$

  The size of the domain of $\delta_{r}$ is $O(|G|\cdot d^{r})$ and a simple BFS algorithm up to distance $r$ for all node of the graphs is enough to compute $\delta_{r}$. When an update occurs, we have to perform the BFS for all $r$-impacted nodes, which can be done in constant time (triply exponential in $\ell$).

  Let $(G,d,\ell)$ fixed. Let $r = 2^{O(\ell)}$, we assume that $BDS(G,d,\ell)$ and the function $\delta_{r}(\cdot,\cdot)$ has been computed.

  We are now given an $\FO$ query $\ph(\x)$ with quantifier rank $q$ and arity $k$ such that $q+k\le \ell$. Similarly to the Naive algorithm, we can assume without loss of generality that $\ph(\x) = \alpha_1(\x_1)\wedge \ldots \wedge \alpha_p(\x_p) \wedge \Delta_{2r,\tau} (\x)$ where $(\x_1,\ldots,\x_p)$ is a partition of $\x$, and all $\alpha_i(\x_i)$ are $r$-locals centered.

  We then compute a final structure in constant time (triply exponential in $|\ph|$) before any tuple is given. The structure is composed of function $SOL_{i,H}(\cdot)$ for all $i\le p$ and $H$ in $\H$. We define and compute these function as follows:

  For all $i\le p$, for all $H$ in $\H$ and for all tuple $\a_i$ (with $|\a_i|=|\x_i|$), we test whether $H\models \alpha_i(\a_i)$. If so, we set $SOL_{i,H}(\a_i) := 1$. Those functions are stored and ready to be accessed in constant time.

  The complexity of this final phase is triply exponential in $\ell$ because there are only $2^{d^{r}}$ graphs in $\H$. For a fixed $H$, there are only $(d^{r})^k$ tuples and the final tests can be performed in time $2^{d^{r}}$.

  We are now given a tuple $\a$ in $G$ and we want to test whether $G\models \ph(\a)$.
  \begin{itemize}
    \item For all $(a_j,a_{j'})$ in $\a$ we check $\delta_{2r}(a_j,a_{j'})$. Therefore in time $O(k^2)$ we know whether $\Delta_{2r,\tau}(\a)$ holds in $G$.
    \item For all $\a_i$ in $\a$, let $a^1_i$ be it's first elements.
    \item Let $H_i = \X(a_i^1)$ and $f_{a_i^1}(\cdot)$ the isomorphism from $\N_{r}^G(a^1_i)$ to $H_i$.

    Since $\Delta_{2r,\tau}(\a)$ holds, we have that $\a_i \subseteq \N^G_{k\cdot r}(a_i^1)$.
    Hence $\N_{r}(\a_i) \subseteq \N_{r\ell}(a^1_i)$.
    Therefore $G \models \alpha_i(\a_i)$ if and only if $H_i\models \alpha_i(f_{a_i^1}(\a_i))$ and this is true if and only if $SOL_{i,H_i}(f_{a_i^1}(\a_i))=1$, which can be tested in time $O(1)$.
  \end{itemize}

  This ends the testing algorithm. The structure can be updated in time triply exponential in $\ell$, and with an additional phase performed once the query is known, we can test whether a given tuple is a solution in time $O(k^2)$.
  \end{proof}
\end{appendix}

\end{document}